\documentclass[12pt,reqno]{article}

\usepackage[usenames]{color}
\usepackage{hyperref}
\usepackage{amssymb}
\usepackage{graphicx}
\usepackage{amscd}

\usepackage{color}
\usepackage{fullpage}
\usepackage{float}

\usepackage{graphics,amsmath,amssymb}
\usepackage{amsthm}
\usepackage{amsfonts}
\usepackage{latexsym}
\usepackage{epsf}

\newcommand{\seqnum}[1]{\href{http://oeis.org/#1}{\underline{#1}}}

\def\Enn{{\mathbb{N}}}

\def\codelink{\url{https://cs.uwaterloo.ca/~shallit/papers.html}}

\begin{document}

\theoremstyle{plain}
\newtheorem{theorem}{Theorem}
\newtheorem{corollary}[theorem]{Corollary}
\newtheorem{lemma}[theorem]{Lemma}
\newtheorem{proposition}[theorem]{Proposition}

\theoremstyle{definition}
\newtheorem{definition}[theorem]{Definition}
\newtheorem{example}[theorem]{Example}
\newtheorem{conjecture}[theorem]{Conjecture}
\newtheorem{openproblem}[theorem]{Open Problem}

\theoremstyle{remark}
\newtheorem{remark}[theorem]{Remark}

\title{Sums of Palindromes: \\
an Approach via Automata}

\author{Aayush Rajasekaran, Jeffrey Shallit, and Tim Smith\\
School of Computer Science \\
University of Waterloo \\
Waterloo, ON  N2L 3G1 \\
Canada \\
{\tt \{arajasekaran,shallit,timsmith\}@uwaterloo.ca}
}

\maketitle

\vskip .2 in
\begin{abstract}
Recently, Cilleruelo, Luca, \& Baxter proved, for all bases $b \geq 5$,  that every natural number is the sum of at most $3$ natural numbers whose base-$b$ representation is a palindrome.  However, the cases $b = 2, 3, 4$ were left unresolved.

We prove, using a decision procedure based on automata,
that every natural number is the sum of at most $4$ natural numbers
whose base-$2$ representation is a palindrome.   Here the constant $4$ is optimal.   We obtain similar results for bases $3$ and $4$, thus completely resolving the problem.

We consider some
other variations on this problem, and prove similar results.  We argue that heavily case-based proofs are a good signal that a decision procedure may help to automate the proof.
\end{abstract}

\section{Introduction}
In this paper we combine three different themes:  (i) additive number
theory; (ii) numbers with special kinds of representations in base $k$;
(iii) use of a decision procedure to prove theorems.  We prove, for example, that every natural number is the sum of at most $9$ numbers whose base-$2$ representation is a palindrome.

Additive number theory is the study of the additive properties of
integers.  
For example,
Lagrange proved (1770) that every natural
number is the sum of four squares \cite{Hardy&Wright:1985}.   In additive number theory, a subset $S \subseteq \Enn$ is called a {\it additive basis of order $h$} if every element of $\Enn$ can be written as a sum of at most
$h$ members of $S$, not necessarily distinct.

Waring's problem asks for the
smallest value $g(k)$ such that the $k$'th powers form a basis of order $g(k)$.
In a variation on Waring's problem, one can ask
for the smallest value $G(k)$ such that every {\it sufficiently large\/}
natural number is the sum of $G(k)$ $k$'th powers \cite{Vaughan&Wooley:2002}.   This kind of representation is called an {\it asymptotic additive basis} of order $G(k)$.

Quoting Nathanson \cite[p.~7]{Nathanson:1996},
\begin{quote}
{\it ``The central problem in additive number theory is to determine if a given set of integers is a basis of finite order.''}
\end{quote}
In this paper we show how to solve this central problem for certain sets, {\it using almost no number theory at all}.

Our second theme concerns numbers with special representations in
base $k$.  For example, numbers of the form $11\cdots 1$ in base
$k$ are sometimes called {\it repunits} \cite{Yates:1978}, and special
effort has been devoted to factoring such numbers, with the Mersenne numbers $2^n - 1$ being the most famous examples.  The {\it Nagell-Ljunggren problem} asks for a characterization of those repunits that are integer powers (see, e.g., \cite{Shorey:1986}).

Another interesting class, and the one that principally concerns us in this article,
consists of those numbers whose
base-$k$ representation forms a {\it palindrome}: a string that
reads the same forwards and backwards, like the English word
{\tt radar}.  Palindromic numbers have
been studied for some time in number theory; see, for example,
\cite{Simmons:1972,
Trigg:1974,
Kresova&Salat:1984,
Aschbacher:1990,
Keith:1990,
Barnett:1991,
Korec:1991,
Harminc&Sotak:1998,
Luca:2003,
Banks&Hart&Sakata:2004,
Hernandez&Luca:2006,
Banks&Shparlinski:2005,
Banks&Shparlinski:2006,
Kidder&Kwong:2007,
Luca&Togbe:2008,
Col:2009,
Cilleruelo&Luca&Shparlinski:2009,
Goins:2009,
Luca&Young:2011,
Basic:2012,
Cilleruelo&Tesoro&Luca:2013,
Berczes&Ziegler:2014,
Basic:2015,
Pollack:2015,
Aloui&Mauduit&Mkaouar:2017}.

Recently Banks 
initiated the study of the additive properties of palindromes,
proving that every natural number is the sum of at most 
49 numbers whose decimal representation is a palindrome \cite{Banks:2016}.  (Also see \cite{Sigg:2015}.)
Banks' result was improved
by Cilleruelo, Luca, \& Baxter \cite{Cilleruelo&Luca:2016,Cilleruelo&Luca&Baxter:2017}, who
proved that for all bases $b \geq 5$, every natural number is
the sum of at most $3$ numbers whose base-$b$ representation is
a palindrome.  The proofs of Banks and 
Cilleruelo, Luca, \& Baxter are both rather lengthy and case-based.
Up to now, there have been no proved results for bases $b = 2,3, 4$.

The long case-based solutions to the problem of representation by sums
of palindromes suggests that perhaps a more automated approach 
might be useful.    For example, in a series of recent papers,
the second author and his co-authors have proved a number of
old and new results in combinatorics on words using a decision procedure based on first-order logic
\cite{Allouche&Rampersad&Shallit:2009,
Charlier&Rampersad&Shallit:2012,
Goc&Henshall&Shallit:2012,
Shallit:2013,
Goc&Saari&Shallit:2013,
Goc&Mousavi&Shallit:2013,
Goc&Schaeffer&Shallit:2013,
Du&Mousavi&Schaeffer&Shallit:2016,
Mousavi&Schaeffer&Shallit:2016,
Du&Mousavi&Rowland&Schaeffer&Shallit:2017}.  The classic result of Thue \cite{Thue:1912,Berstel:1995} that the Thue-Morse infinite word  ${\bf t} = 0110100110010110 \cdots$ avoids overlaps (that is, blocks of the form $axaxa$ where $a$ is a single letter and $x$ is a possibly empty block) is an example of a case-based proof that can be entirely replaced \cite{Allouche&Rampersad&Shallit:2009} with a decision procedure based on the first-order logical theory ${\rm FO}(\Enn, +,  V_2)$.

Inspired by these and other successes in automated deduction and
theorem-proving (e.g., \cite{McCune:1996}),
we turn to formal languages and automata
theory as a suitable framework for expressing the palindrome representation
problem.  Since we want to make assertions about the representations of
{\it all\/} natural numbers, this requires finding (a) a machine model or
logical theory in which universality is decidable and (b) a variant of
the additive problem of palindromes suitable for this machine model or
logical theory.   The first model we use is the {\it nested-word automaton}, a variant of the more familiar pushdown automaton.  This is used to handle the case for base $b = 2$.  The second model we use is the ordinary finite automaton.  We use to resolve the cases $b = 3,4$.

Our paper is organized as follows:
In Section~\ref{sopp} we introduce some notation and terminology, and
state more precisely the problem we want to solve.  In Section~\ref{nwa} we recall the pushdown automaton model and give an example, and we motivate our use of nested-word automata.
In Section~\ref{solve} we restate our problem in the framework of nested-word automata, and the proof of a bound of $4$ palindromes is given in Section~\ref{nwap-proof}.  The novelty of our approach involves replacing the long case-based reasoning of previous proofs with an automaton-based approach using a decision procedure.   In Section~\ref{variations} we consider some
variations on the original problem.   In Section~\ref{objections} we discuss possible objections to our approach.  In Section~\ref{future} we describe future work.   Finally, in Section~\ref{moral}, we conclude our paper by stating a thesis underlying our approach.

\section{The sum-of-palindromes problem}
\label{sopp}

We first introduce some notation and terminology.  

The natural numbers are $\Enn = \{ 0,1,2,\ldots\}$.
If $n$ is a natural number, then by $(n)_k$ we mean the string (or word)
representing $n$ in base $k$, with no leading zeroes, starting with
the most significant digit.  Thus, for
example, $(43)_2 = 101011$.  The alphabet $\Sigma_k$ is defined
to be $\{ 0,1,\ldots, k-1 \}$; by $\Sigma_k^*$ we mean the set of all
finite strings over $\Sigma_k$.  If $x \in \Sigma_\ell^*$ for some $\ell$, then by
$[x]_k$ we mean the integer represented by the string $x$, considered
as if it were a number in base $k$, with the most significant digit at the left.  That is,
if $x = a_1 a_2 \cdots a_n$, then
$[x]_k = \sum_{1 \leq i \leq n} a_i k^{n-i}$.
For example, $[135]_2 = 15$.

If $x$ is a string, then $x^i$ denotes the string
$\overbrace{xx\cdots x}^i$, and $x^R$ denotes the reverse of $x$.  Thus,
for example, $({\tt ma})^2 = {\tt mama}$, and $({\tt drawer})^R = {\tt reward}$.   If $x = x^R$, then
$x$ is said to be a {\it palindrome}.   

We are interested in integers whose base-$k$ representations are
palindromes.  In this article, we routinely abuse terminology by calling
such an integer a {\it base-$k$ palindrome}.  In the case where $k = 2$,
we also call such an integer a {\it binary palindrome}.  The first
few binary palindromes are
$$ 0,1,3,5,7,9,15,17,21,27,31,33,45,51,63, \ldots ;  $$
these form sequence \seqnum{A006995} in the
{\it On-Line Encyclopedia of Integer Sequences} (OEIS).

If $k^{n-1} \leq r < k^n$ for $n \geq 1$, we say that
$r$ is an {\it $n$-bit integer\/} in base $k$.  If $k$ is unspecified,
we assume that $k = 2$.  Note that the first bit of an $n$-bit integer
is always nonzero.  The {\it length\/} of an integer $r$ satisfying
$k^{n-1} \leq r < k^n$ is defined to be $n$; alternatively, the
length of $r$ is $1 + \lfloor \log_k r \rfloor$.

Our goal is to find a constant $c$ such that every natural number is the sum of at most $c$ binary palindromes.  To the best of our knowledge, no such bound has been proved up to now.
In Sections~\ref{solve} and \ref{nwap-proof} we describe how we used a decision procedure for
nested-word automata to prove the following result:

\begin{theorem}
For all $n \geq 8$, every $n$-bit odd integer
is either a binary palindrome itself, or the sum of three binary palindromes
\begin{itemize}
\item[(a)]  of lengths $n$, $n-2$, and $n-3$; or
\item[(b)]   of lengths $n-1$, $n-2$, and $n-3$.
\end{itemize}
\label{nwap}
\end{theorem}

As a corollary, we get our main result:

\begin{corollary}
Every natural number $N$ is the sum of at most $4$ binary palindromes.
\label{main}
\end{corollary}

\begin{proof}
It is a routine computation to verify the result for $N < 128$.  

Now suppose $N \geq 128$.  Let $N$ be an $n$-bit integer; then 
$n \geq 8$.  If $N$ is odd, then
Theorem~\ref{nwap} states that $N$ is the sum of at most $3$ binary palindromes.  Otherwise, $N$ is even.

If $N = 2^{n-1}$, then it is the sum of $2^{n-1} -1$ and
$1$, both of which are palindromes.

Otherwise, $N-1$ is also an $n$-bit odd integer.
Use Theorem~\ref{nwap} to find a representation for $N-1$
as the sum of at most $3$ binary palindromes, and
then add the palindrome $1$ to get a representation for $N$.
\end{proof}

\begin{remark}
We note that the bound $4$ is optimal since, for example, the number $176$ is not the sum of three or fewer binary palindromes.

Sequence \seqnum{A261678} in the OEIS lists those even numbers that are not the sum of two binary palindromes.  Sequence \seqnum{A261680} gives the number of distinct representations as the sum of four binary palindromes.
\end{remark}

\section{Finding an appropriate computational model}
\label{nwa}

To find a suitable model for proving Theorem~\ref{nwap}, we turn to formal languages and automata.  We seek some class of automata with the following property: for each $k$, there is an automaton which, given a natural number $n$ as input, accepts the input iff $n$ can be expressed as the sum of $k$ palindromes.  Furthermore, we would like the problem of universality (``Does the automaton accept every possible input?'') to be decidable in our chosen model.  By constructing the appropriate automaton and checking whether it is universal, we could then determine whether every number $n$ can be expressed as the sum of $k$ palindromes.

Palindromes suggest considering
the model of pushdown automaton (PDA), since it is
well-known that this class of machines, equipped with a stack, can accept the palindrome language ${\tt PAL} = \{ x \in \Sigma^* \ : \ x = x^R \}$ over any fixed alphabet $\Sigma$. 
A tentative approach is as follows: create a PDA $M$ that, on input $n$ expressed in base $2$, uses nondeterminism to ``guess'' the $k$ summands and verify that (a) every summand is a palindrome, and (2) they sum to the input $n$.  We would then check to see if $M$ accepts all of its inputs.  However, two problems immediately arise.  

The first problem is that universality is recursively unsolvable for nondeterministic PDAs \cite[Thm.~8.11, p.~203]{Hopcroft&Ullman:1979}, so even if the automaton $M$ existed, there would be no algorithm guaranteed to check universality.   

The second problem involves checking that the guessed summands are palindromes.  One can imagine guessing the summands in parallel, or in series.  If we try to check them in parallel, this seems to correspond to the recognition of a language which is not a CFL (i.e., a context-free language, the class of languages recognized by nondeterministic PDAs).  Specifically, we encounter the following obstacle:

\begin{theorem}
The set of strings $L$ over the alphabet $\Sigma\times (\Sigma \ \cup \ \#)$,
where the first ``track''
is a palindrome and the second ``track'' is a another, possibly shorter, palindrome, padded on the right with \# signs, is not a CFL.
\end{theorem}

\begin{proof}
Assume that it is.  Consider $L$ intersected with the regular language 
$$[1,1]^+ [1,0] [1,1]^+ [0,1] [1,\#]^+,$$
and call the result $L'$.  We use Ogden's lemma \cite{Ogden:1968} to show $L'$ is not a CFL.

Choose $z$ to be a string where the first track is $(1^{2n} 0 1^{2n})$ and the second track is $(1^n 0 1^n \#^{2n})$.  Mark the compound symbols $[1,\#]$.
Then every factorization $z = uvwxy$ must have at least one $[1,\#]$ in $v$ or $x$.
If it's in $v$, then the only choice for $x$ is also $[1,\#]$, so pumping gives a non-palindrome on the first track.
If it's in $x$ then $v$ can be $[1,1]^i$ or contain $[1,0]$ or $[0,1]$ .  If the latter, pumping twice gives a string not in $L'$ because there is more than one $0$ on one of the two tracks.  If the former, pumping twice gives a string with the second track not a palindrome.   This contradiction shows that $L'$, and hence $L$, is not a context-free language.
\end{proof}
So, using a pushdown automaton, we cannot check arbitrary palindromes of wildly unequal lengths in parallel.

\medskip

If the summands were presented serially, we could check whether each summand individually is a palindrome, using the stack, but doing so destroys our copy of the summand, and so we cannot add them all up and compare them to the input.  In fact, we cannot add serial summands in any case, because we have 
\begin{theorem}
The language 
$$ L = \{ (m)_2 \# (n)_2 \# (m+n)_2 \ : \  m, n \geq 0 \} $$
is not a CFL. 
\end{theorem}

\begin{proof}
Assume $L$ is a CFL and intersect with the regular language
$1^+ 0 1^+ \# 1^+ \# 1^+ 0$, obtaining $L'$.  We claim that
$$L' = \{ 1^a 0 1^b \# 1^c \# 1^d 0 \ : \ b=c \text{ and } a+b = d \}.$$
This amounts to the claim, easily verified, that 
the only solutions to the equation $2^{a+b+1} - 2^b - 1 + 2^c - 1 = 2^{d+1}- 2$ are $b=c$ and $a+b  = d$.  Then, starting with the
string $z = 1^n 0 1^n \# 1^n \# 1^{2n} 0 $,  an easy argument with Ogden's lemma proves that $L'$ is not a CFL, and hence neither is $L$.
\end{proof}
So, using a pushdown automaton, we can't handle summands given in series, either.

\medskip

These issues lead us to restrict our attention to representations as sums of palindromes of the same (or similar) lengths.   More precisely, we consider the following variant of the additive problem of palindromes: for a length $l$ and number of summands $k$, given a natural number $n$ as input, is $n$ the sum of $k$ palindromes all of length exactly $l$?  Since the palindromes are all of the same length, a stack would allow us to guess and verify them in parallel.  To tackle this problem, we need a model which is both (1) powerful enough to handle our new variant, and (2) restricted enough that universality is decidable.  We find such a model in the class of {\it nested-word automata}, described in the next section.

\section{Restating the problem in the language of nested-word automata}
\label{solve}

Nested-word automata (NWAs) were popularized by Alur and Madhusudan  \cite{Alur&Madhusudan:2004,Alur&Madhusudan:2009}, although essentially the same model was discussed previously  by  Mehlhorn \cite{Mehlhorn:1980}, von Braunm{\"u}hl and Verbeek \cite{Braunmuhl&Verbeek:1983}, and Dymond \cite{Dymond:1988}.
They are a restricted variant of pushdown automata.  Readers familiar with visibly pushdown automata (VPA) should note that NWAs are an equally powerful machine model as VPAs \cite{Alur&Madhusudan:2004, Alur&Madhusudan:2009}.  We only briefly describe their functionality here.   For other theoretical aspects of nested-word and visibly-pushdown automata, see  \cite{LaTorre&Napoli&Parente:2006,Piao&Salomaa:2009,Han&Salomaa:2009,Salomaa:2011,Okhotin&Salomaa:2017}.

The input alphabet of an NWA is partitioned into three sets: a \emph{call alphabet}, an \emph{internal alphabet}, and a \emph{return alphabet}.  An NWA has a stack, but has more restricted access to it than PDAs do. If an input symbol is from the internal alphabet, the NWA cannot access the stack in any way.  If the input symbol read is from the call alphabet, the NWA pushes its current state onto the stack, and then performs a transition, based only on the current state and input symbol read.  If the input symbol read is from the return alphabet, the NWA pops the state at the top of the stack, and then performs a transition based on three pieces of information: the current state, the popped state, and the input state read. An NWA accepts if the state it terminates in is an accepting state.

As an example, Figure~\ref{fig1} illustrates a nested-word automaton accepting the language $\{ 0^n 1 2^n \ : \ n \geq 1 \}$.
\begin{figure}[H]
\begin{center}
\includegraphics[width=14cm]{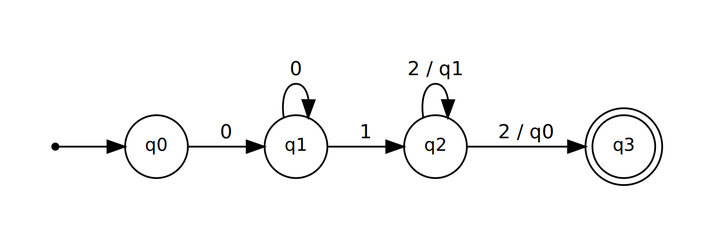}
\end{center}
\caption{A nested-word automaton for the language $\{ 0^n 1 2^n \ : \ n \geq 1 \}$}
\label{fig1}
\end{figure}
\noindent Here the call alphabet is $\{0\}$, the internal alphabet is $\{1\}$, and the return alphabet is $\{2\}$.

Nondeterministic NWAs are a good machine model for our problem, because nondeterminism allows ``guessing'' the palindromes that
might sum to the input, and the stack allows us to ``verify" that they are indeed palindromes. Deterministic NWAs are as expressive as nondeterministic NWAs, and the class of languages they accept is closed under the operations of union, complement and intersection. Finally, testing emptiness, universality, and language inclusion are all decidable problems for NWAs \cite{Alur&Madhusudan:2004,Alur&Madhusudan:2009}; there is an algorithm for all of them.

For a nondeterministic NWA of $n$ states, the corresponding determinized machine has at most $2^{\Theta(n^2)}$ states, and there are examples for which this bound is attained. This very rapid explosion in state complexity potentially could make deciding problems such as language inclusion infeasible in practice. Fortunately, we did not run into determinized machines with more than $40000$ states in proving our results. Most of the algorithms invoked to prove our results run in under a minute. 

We now discuss the general construction of the NWAs that check whether inputs are sums of binary palindromes. We partition the input alphabet into the call alphabet $\{a,b\}$, the internal alphabet $\{c,d\}$, and the return alphabet $\{e,f\}$. The symbols $a$, $c$, and $e$ correspond to 0, while $b$, $d$, and $f$ correspond to 1. The input string is fed to the machine starting with the {\it least significant digit}. We provide the NWA with input strings whose first half is entirely made of call symbols, and second half is entirely made of return symbols. Internal symbols are used to create a divider between the halves (for the case of odd-length inputs).

The idea behind the NWA is to nondeterministically guess all possible summands when reading the first half of the input string. The guessed summands are characterized by the states pushed onto the stack. The machine then checks if the guessed summands can produce the input bits in the second half of the string. The machine keeps track of any carries in the current state. 

To check the correctness of our NWAs, we built an NWA-simulator, and then ran simulations of the machines on various types of inputs, which we then checked against experimental results. 

For instance, we built the machine that accepts representations of integers that can be expressed as the sum of 2 binary palindromes. We then simulated this machine on every integer from 513 to 1024, and checked that it only accepts those integers that we experimentally confirmed as being the sums of 2 binary palindromes.

The general procedure to prove our results is to build an NWA {\tt PalSum} accepting only those inputs that it verifies as being appropriate sums of palindromes, as well as an NWA {\tt SyntaxChecker} accepting all valid representations. We then run the decision algorithms for language inclusion, language emptiness, etc. on {\tt PalSum} and {\tt SyntaxChecker} as needed. To do this, we used the Automata Library toolchain of the ULTIMATE program analysis framework \cite{Heizmann&co:2013,Heizmann&co:2016}. 

We have provided links to the proof scripts used to establish all of our results. To run these proof scripts, simply copy the contents of the script into
\url{https://monteverdi.informatik.uni-freiburg.de/tomcat/Website/?ui=int&tool=automata_library} and click ``execute".

\section{Proving Theorem~\ref{nwap}}
\label{nwap-proof}

In this section, we discuss construction of the appropriate nested-word automaton in more detail.   

\begin{proof} (of Theorem~\ref{nwap})

We build three separate automata. The first, {\tt palChecker}, has 9 states, and simply checks whether the input number is a binary palindrome. The second, {\tt palChecker2}, has 771 states, and checks whether an input number of length $n$ can be expressed as the sum of three binary palindromes of lengths $n$, $n-2$, and $n-3$. The third machine, {\tt palChecker3}, has 1539 states, and checks whether an input number of length $n$ can be expressed as the sum of three binary palindromes of lengths $n-1$, $n-2$, and $n-3$. We then determinize these three machines, and take their union, to get a single deterministic NWA, {\tt FinalAut}, with 36194 states. 

The language of valid inputs to our automata is given by $$ L = \{ \{a,b\}^n  \{c,d\}^m \{e,f\}^n \ :  \ 0 \leq m \leq 1, \ n \geq 4
\} .$$

We only detail the mechanism of {\tt palChecker3} here. Let $p,q$ and $r$ be the binary palindromes representing the guessed $n$-length summand, $(n-2)$-length summand and $(n-3)$-length summand respectively.
The states of {\tt palChecker3} include 1536 $t$-states that are 10-tuples. We label these states $(g, x, y,z,k,l_1,l_2,m_1,m_2,m_3)$, where $0 \leq g \leq 2$, while all other coordinates are either 0 or 1. The $g$-coordinate indicates the current carry, and can be as large as 2. The $x, y$ and $z$ coordinates indicate whether we are going to guess $0$ or $1$ for the next guesses of $p,q$ and $r$ respectively. The remaining coordinates serve as ``memory'' to help us manage the differences in lengths of the guessed summands. The $k$-coordinate remembers the most recent guess for $p$. We have $l_1$ and $l_2$ remember the two most recent $q$ guesses, with $l_1$ being the most recent one, and we have $m_1, m_2$ and $m_3$ remember the three most recent $r$-guesses, with $m_1$ being the most recent one, then $m_2$, and $m_3$ being the guess we made three steps ago. We also have three $s$-states labeled $s_0, s_1$ and $s_2$, representing carries of 0, 1 and 2 respectively. These states process the second half of the input string.

The initial state of the machine is $(0,1,1,1,0,0,0,0,0,0)$ since we start with no carry, must guess 1 for our first guess of a valid binary palindrome, and all ``previous'' guesses are 0.
A $t$-state has an outgoing transition on either $a$ or $b$, but not both. If $g+x+y+z$ produces an output bit of 0, it takes a transition on $a$, else it takes a transition on $b$.
The destination states are all six states of the form $(g', x', y',z',x,y,l_1,z,m_1,m_2)$, where $g'$ is the carry resulting from $g+x+y+z$, and $x',y',z'$ can be either 0 or 1.
Note that we ``update'' the remembered states by forgetting $k, l_2$ and $m_3$, and saving $x,y$ and $z$.

The $s$-states only have transitions on the return symbols $e$ and $f$. When we read these symbols, we pop a $t$-state off the stack. If state $s_i$ pops the state $(g, x, y,z,k,l_1,l_2,m_1,m_2,m_3)$ off the stack, its transition depends on the addition of $i+k+l_2+m_3$. If this addition produces a carry of $j$, then $s_i$ can take a transition to $s_j$ on $e$ if the output bit produced is 0, and on $f$ otherwise. 
By reading the $k$, $l_2$ and $m_3$ values, we correctly realign $p$, $q$ and $r$, correcting for their different lengths. This also ensures that we supply 0s for the last three guesses of $r$, the last two guesses of $q$ and the last guess of $p$.
The only accepting state is $s_0$.

It remains to describe how we transition from $t$-states to $s$-states. This transition happens when we are halfway done reading the input. If the length of the input is odd, we read a $c$ or a $d$ at the halfway point. We only allow certain $t$-states to have outgoing transitions on $c$ and $d$. Specifically, we require the state's coordinates to satisfy $k=x$, $l_2=y$, $m_1=m_2$ and $m_3=z$. These conditions are required for $p$, $q$ and $r$ to be palindromes. We transition to state $s_i$ if the carry produced by adding $g+x+y+z$ is $i$, and we label the transition $c$ if the output bit is 0, and $d$ otherwise.

If the length of the input is even, then our $t$-states take a return transition on $e$ or $f$. Once again, we restrict the $t$-states that can take return transitions. We require the state's coordinates to satisfy $l_1 = l_2$ and $m_1 = m_3$ to ensure our guessed summands are palindromes. Let the current state be $(g, x, y,z,k,l_1,l_1,m_1,m_2,m_1)$, and the state at the top of the stack be $(g', x', y',z',k',l_1',l_2',m_1',m_2',m_3')$. We can take a transition to $s_i$ if the sum $g+k'+l_2'+m_3'$ produces a carry of $i$, and we label the transition $e$ if the output bit is 0, and $f$ otherwise. 

The structure and behavior of {\tt palChecker2} is very similar. One difference is that there is no need for a $k$-coordinate in the $t$-states since the longest summand guessed is of the same length as the input.

The complete script executing this proof is over 750000 lines long. Since these automata are very large, we wrote two C\texttt{++} programs to generate them. Both the proof script, and the programs generating them can be found at \codelink.
It is worth noting that a $t$-state labeled as $(g, x, y,z,k,l_1,l_2,m_1,m_2,m_3)$ in this report is labeled $ q\_g\_xyz\_k\_l_1 l_2\_m_1 m_2 m_3$ in the proof script. Also, ULTIMATE does not currently have a union operation for NWAs, so we work around this by using De Morgan's laws for complement and intersection.
\end{proof}

\section{Bases 3 and 4}
\label{larger-bases}

In this section we prove analogous results for bases 3 and 4. We show that every natural number
is the sum of at most three base-3 palindromes, and at most three base-4 palindromes.  Because the NWAs needed became too large for us to manipulate effectively, we
use a modified approach using nondeterministic 
finite automata to prove these results,

Our result for base 3 is as follows:

\begin{theorem}
For all $n \geq 9$, every integer whose base-3 representation is of length $n$
is the sum of 
\begin{itemize}
\item[(a)]  three base-3 palindromes of lengths $n$, $n-1$, and $n-2$; or
\item[(b)]  three base-3 palindromes of lengths $n$, $n-2$, and $n-3$; or
\item[(c)]  three base-3 palindromes of lengths $n-1$, $n-2$, and $n-3$; or
\item[(d)]  two base-3 palindromes of lengths $n-1$ and $n-2$.
\end{itemize}
\label{base3}
\end{theorem}

\begin{proof}
We represent the input in a ``folded'' manner over the input alphabet
$\Sigma_3 \cup (\Sigma_3 \times \Sigma_3)$, where
$\Sigma_k = \{ 0,1, \ldots, k-1 \}$, giving the machine 2
letters at a time from opposite ends. This way we can directly guess
our summands without having need of a stack at all. We align the input
along the $n-2$ summand by providing the first 2 letters of the input separately.

If $(N)_3 = a_{2i+1}a_{2i}\cdots a_0$,
we represent $N$ as the word 
$$a_{2i+1}a_{2i}[a_{2i-1}, a_{i-1}][a_{2i-2}, a_{i-2}]
\cdots [a_{i}, a_{0}]$$.

Odd-length inputs leave a trailing unfolded letter at the end of their input. 
If $(N)_3 = a_{2i+2}a_{2i+1}\cdots a_0$,
we represent $N$ as the word 
$$a_{2i+2}a_{2i+1}[a_{2i}, a_{i-1}][a_{2i-1}, a_{i-2}]
\cdots [a_{i+1}, a_{0}]a_{i}$$.

We need to simultaneously carry out addition on both ends. 
In order to do this we need to keep track of two carries. 
On the lower end, we track the ``incoming'' carry at the start of 
an addition, as we did in the proofs using NWAs. On the higher end,
however, we track the expected ``outgoing'' carry. 

To illustrate how our machines work, we consider an NFA accepting 
length-$n$ inputs that are the sum of 4 base-3 palindromes, one each
of lengths $n$, $n-1$, $n-2$ and $n-3$. Although this is not a case
in our theorem, each of the four cases in our theorem can be obtained
from this machine by striking out one or more of the guessed summands.

Recall that we aligned our input along the length-$(n-2)$ summand by
providing the 2 most significant letters in an unfolded manner.
This means that our guessed for the length-$n$ summand 
will be ``off-by-two'': when we make a guess at the higher end of the
length-$n$ palindromic summand, its appearance at the lower end 
is 2 steps away. We hence need to remember the last 2 guesses at 
the higher end of the length-$n$ summand in our state. Similarly, 
we need to remember the most recent higher guess of the length-$(n-1)$
summand, since it is off-by-one. The length-$(n-2)$ summand is perfectly
aligned, and hence nothing needs to be remembered. The length-$(n-3)$ summand
has the opposite problem of the length-$(n-1)$ input. Its lower guess only
appears at the higher end one step later, and so we save the most recent
guess at the lower end. 

Thus, in this machine, we keep track of 6 pieces of information: 

\begin{itemize}
\item $c_1$, the carry we are expected to produce on the higher end,
\item$c_2$, the carry we have entering the lower end,
\item $x_1$ and $x_2$, the most recent higher guesses of the length-$n$ summand,
\item $y$, the most recent higher guess of the length-$(n-1)$ summand, and
\item $z$, the most recent lower guess of the length-$(n-3)$ summand,
\end{itemize}

Consider a state $(c_1, c_2, x_1, x_2, y, z)$. Let $0 \leq i, j, k, l < 3$ be our
next guesses for the four summands of lengths $n$, $n-1$, $n-2$ and $n-3$ respectively.
Also, let $\alpha$ be our guess for the next incoming carry on the higher end.
Let the result of adding $i+j+k+z+\alpha$ be a value $0 \leq p_1 < 3$ and a carry of
$q_1$.
Let the result of adding $x+y+k+l+x_2$ be a value $0 \leq p_2 < 3$ and a carry of $q_2$.
We must have $q_1 = c_1$. If this condition is met, we add a transition from this state
to $(\alpha, q_2, x_2, i, j, l)$, and label the transition $[p_1, p_2]$.

The initial state is $(0,0,0,0,0,0)$. We expand the alphabet to include 
special variables for the first 3 symbols of the input string. This is to
ensure that we always guess a $1$ or a $2$ for the first (and last) positions of our
summands. 

The acceptance conditions depend on whether $(N)_3$ is of even or odd length.
If a state $(c_1, c_2, x_1, x_2, y, z)$ satisfies $c_1 = c_2$ and 
$x_1 = x_2$, we set it as an accepting states. 
A run can only terminate in one of these states if $(N)_3$ is of even length.
We accept since we are confident that our guessed summands are palindromes 
(the condition $x_1 = x_2$ ensures our length-$n$ summand is palindromic), and
since the last outgoing carry on the lower end is the expected first 
incoming carry on the higher end (enforced by $c_1=c_2$).

We also have a special symbol to indicate the trailing symbol of an
input for which $(N)_3$ is of odd length. We add transitions from our states to 
a special accepting state, $q_acc$, if we read this special symbol.
Consider a state $(c_1, c_2, x_1, x_2, y, z)$, and let $0 \leq k < 3$ 
be our middle guess for the $n-2$ summand. Let the result of 
adding $x_1+y+k+z+c_2$ be a value $0 \leq p < 3$ and a carry of
$q$. If $q=c_1$, we add a transition on $p$ from our state to $q_acc$.

We wrote a C\texttt{++} program generating the machines for each case of this theorem.
After taking the union of these machines, determinizing and minimizing, the machine
has 378 states. We then built a second NFA that accepts folded representations of
$(N)_3$ such that the unfolded length of $(N)_3$ is greater than 8. We then use
ULTIMATE to assert that the language accepted by the second NFA is included in 
that accepted by the first. All these operations run in under a minute.

We tested this machine by experimentally calculating which values of 
$243 \leq N \leq 1000$ could be written as the sum of palindromes satisfying
one of our 4 conditions. We then asserted that for all the folded representations
of $243 \leq N \leq 1000$, our machine accepts these values which we
experimentally calculated, and rejects all others.

\end{proof}

We also have the following result for base 4:

\begin{theorem}
For all $n \geq 7$, every integer whose base-4 representation is of length $n$
is the sum of 
\begin{itemize}
\item[(a)]  exactly one palindrome each of lengths $n-1$, $n-2$, and $n-3$; or
\item[(b)]  exactly one palindrome each of lengths $n$, $n-2$, and $n-3$.
\end{itemize}
\label{base4}
\end{theorem}

\begin{proof}
The NFA we build is very similar to the machine described for the base-3 proof.
Indeed, the generator used is the same as the one for the base-3 proof, 
except that its input base is 4, and the only machines it generates are
for the two cases of this theorem. The minimized machine has 478 states.
\end{proof}

This, together with the results previously obtained by Cilleruelo, Luca, \& Baxter, completes the additive theory of palindromes for all integer bases $b \geq 2$.

\section{Variations on the original problem}
\label{variations}

In this section we consider some variations on the original problem.
Our first variation involves the notion of generalized palindromes.

\subsection{Generalized palindromes}
\label{genpal}

We define a {\it generalized palindrome of length $n$} to be an integer whose base-$k$ representation, extended, if necessary, by adding leading zeroes to make it of length $n$, is a palindrome.
For example, $12$ is a generalized binary palindrome of length $6$, since its length-$6$ representation $001100$ is a palindrome.  If a number is a generalized palindrome of any length, we call it a {\it generalized palindrome}.
The first few binary generalized palindromes are
$$ 0,1,2,3,4,5,6,7,8,9,10,12,14,15,16,17,18,20,21,24,27,28,30,31,32, \ldots ;$$
they form sequence \seqnum{A057890} in the OEIS. 

We can use our method to prove

\begin{theorem}
Every natural number is the sum of at most $3$ generalized binary palindromes.
\end{theorem}

\begin{proof}
For $N < 32$, the result can easily be checked by hand. 

We then use our method to prove the following claim:

Every length-$n$ integer, $n \geq 6$,
is the sum of at most two generalized
palindromes of length $n$ and one of length $n-1$. 

We build a single NWA, {\tt gpalChecker}, accepting representations of length-$n$ integers that can be written as the sum of at most two generalized
palindromes of length $n$ and one of length $n-1$. The structure of {\tt gpalChecker} is a highly simplified version of
the one described in Section ~\ref{nwap-proof}. With generalized palindromes we allow the first guessed bit to be zero. We do not have to worry about the ``at most'' clause in our claim, since the string of $n$ 0s is a valid generalized palindrome. The nondeterministic machine has 39 states, which grows to 832 states on determinization. As before, we build an NWA {\tt syntaxChecker} that accepts all valid inputs, and assert that the language accepted by {\tt syntaxChecker} is included in that accepted by {\tt gpalChecker}.

The complete proof script, as well as the C\texttt{++} program generating it can be found at \codelink. 

\end{proof}

As an aside, we also mention the following enumeration result.

\begin{theorem}
There are exactly $3^{\lceil n/2\rceil}$ natural numbers
that are the sum of two generalized binary palindromes of the same length $n$.
\end{theorem}

\begin{proof}
Take two generalized binary palindromes of length $n$ and add them
bit-by-bit.  Then each digit position is either $0, 1$, or $2$, and the
result is still a palindrome.  Hence there are at most
 $3^{\lceil n/2\rceil}$ such numbers.  It remains to see they are all
 distinct.

We show that if $x$ is a length-$n$ word that is a palindrome over
$\{0,1,2\}$, then the map that sends $x$ to its evaluation in base 2
is injective.  

We do this by induction on the length of the palindrome.  The claim is easily verified for length $0$ and $1$.  Now suppose $x$ and $y$ are two distinct strings of the
same length that evaluate to the same number in base 2.   Then, by
considering everything mod 2, we see that either
\begin{itemize}
\item[(a)] $x$ and $y$ both end in the same number $i$, or
\item[(b)] (say) $x$ ends in 0 and $y$ ends in 2.
\end{itemize}

In case (a) we know that $x$ and $y$, since they are both palindromes,
both begin and end in $i$.  So subtracting $i00\cdots 0i$ from $x$ and $y$
we get words $x'$, $y'$ of length $n-2$ that evaluate to the
same thing, and we can apply induction to see $x' = y'$.

In case (b), say $x = 0 x' 0$ and $y = 2 y'2$, and they evaluate to the
same thing.   However, the largest $x$ could be is $0222\cdots 20$
and the smallest $y$ could be is $200\cdots 02$.  The difference between these two numbers
is 6, so we conclude that this is impossible.
\end{proof}

\begin{remark} 
Not every natural number is the sum of 2 generalized binary palindromes; the smallest exception is $157441$.  The list of exceptions forms sequence  \seqnum{A278862} in the OEIS.  It does not seem to be known currently whether there are infinitely many exceptions. 
\end{remark}

\subsection{Antipalindromes}

Another generalization concerns antipalindromes.   Define a map from $\Sigma_2^*$ to $\Sigma_2^*$ by changing every $0$ to $1$ and vice-versa; this is denoted $\overline{w}$.  Thus, for example,
$\overline{010} = 101$.  An {\it antipalindrome} is a word $w$ such that $w = \overline{w^R}$.  It is easy to see that if a word $w$ is an antipalindrome, then it must be of even length.
We call an integer $n$ an antipalindrome if its canonical base-$2$ expansion is an antipalindrome.  The first few antipalindromic integers are
$$ 2,10,12,38,42,52,56,142,150,170,178,204,212,232,240, \ldots ;$$
they form sequence \seqnum{A035928} in the OEIS.
Evidently every antipalindrome must be even.

To use NWAs to prove results about antipalindromes, we slightly modify our machine structure. While processing the second portion of the input string, we complement the number of ones we guessed in the first portion. Assume we are working with $n$ summands, and let the input symbol be $q \in \{e,f\}$, and the state at the top of the stack be $f_{x,y}$. The state $s_i$ has a transition to $s_j$ if $i+(n-y)$ produces an output bit corresponding to $q$ and a new carry of $j$. We also set the initial state to be $f_{0,0}$, because the least significant bit of an antipalindrome must be 0.

Using our method, we would like to prove the following conjecture.  However, so far we have been unable to complete the computations because our program runs out of space.  In a future version of this paper, we hope to change the status of the following conjecture to a theorem.

\begin{conjecture}
Every even integer of length $n$, for $n$ odd and $n \geq 9$,
is the sum of at most $10$ antipalindromes of length $n-3$. 
\label{antipal}
\end{conjecture}

\begin{corollary} (Conditional only on the truth of Conjecture~\ref{antipal})
Every even natural number $N$ is the sum of at most $13$ antipalindromes.
\label{antip}
\end{corollary}

\begin{proof}
For $N < 256$ this can be verified by a short program.
(In fact, for these $n$, only 4 antipalindromes are needed.)  Otherwise assume $N \geq 256$.  Then $N$ is of length $n \geq 9$.  If $n$ is odd, we are done by
Theorem~\ref{antipal}.  Otherwise, $n$ is even.
Then, by subtracting at most $3$ copies of
the antipalindrome $2^{n-2} - 2^{{n \over 2}-1}$ from $N$, we obtain $N'$ even, of  length $n-1$.
We can then apply Theorem~\ref{antipal} to $N'$.
\end{proof}

The $13$ in Corollary~\ref{antip} is probably quite far from the optimal bound.  Numerical evidence suggests

\begin{conjecture}
Every even natural number is the sum of at most $4$ antipalindromes.
\end{conjecture}

\begin{conjecture}
Every even natural number except
\begin{multline*} 8,18,28,130,134,138,148,158,176,318,530,538,548,576,644,1300,2170,2202,2212,2228, \\
2230,2248,8706,8938,8948,34970,35082
\end{multline*}
is the sum of at most $3$ antipalindromes.
\end{conjecture}

\subsection{Generalized antipalindromes}

Similarly, one could consider ``generalized antipalindromes''; these are numbers whose base-$2$ expansion becomes an antipalindrome if a suitable number of leading zeroes are added.  The notion of length here is the same as in Section~\ref{genpal}.

\begin{theorem}

Every natural number of length $n$, for $n \geq 6$ and even, is the sum of exactly $6$ generalized antipalindromes of length $n-2$.

\label{gap2}
\end{theorem}

\begin{proof}
Since generalized antipalindromes can have leading zeroes, we allow all $f$-states with no carry as initial states. We also complement the number of ones for the second half, as we did when handling regular antipalindromes.

The automated proof of (a) can be found at \codelink. The determinized automaton has 2254 states.

\end{proof}

\begin{corollary}
Every natural number is the sum of at most
$7$ generalized antipalindromes.
\label{gap3}
\end{corollary}

\begin{proof}
Just like the proof of Corollary~\ref{main},
using Theorem~\ref{gap2}.
\end{proof}

\begin{remark}
Corollary~\ref{gap3} is probably not best possible.  The correct bound seems to be $3$.
The first few numbers that do not have a representation as the sum of $2$ generalized antipalindromes are
\begin{multline*}
29,60,91,109,111,121,122,131,135,272,329,347,365,371,373,391,401,429,441,445,449,469, \\
473,509,531,539,546,577,611,660,696,731,744,791,804,884,905,940,985,1011,1020,1045, \ldots
\end{multline*}
\end{remark}

\section{Objections to this kind of proof}
\label{objections}

A proof based on computer calculations, like the one we have presented here, is occasionally criticized because it cannot easily be verified by hand, and because it relies on software that has not been formally proved. These kinds of criticisms are not new; they date at least to the 1970's, in response to the celebrated proof of the four-color theorem by Appel and Haken \cite{Appel&Haken:1977,Appel&Haken&Koch:1977}.  See, for example, Tymoczko \cite{Tymoczko:1979}.

We answer this criticism in several ways.  First, it is not reasonable to expect that every result of interest to mathematicians will have short and simple proofs.  There may well be, for example, easily-stated results for which the shortest proof possible in a given axiom system is longer than any human mathematician could verify in their lifetime, even if every waking hour were devoted to checking it.  For these kinds of results, an automated checker may be our only hope.  There are many results for which the only proof currently known is computational.

Second, while short proofs can easily be checked by hand, what guarantee is there that any very long case-based proof --- whether constructed by humans or computers --- can always be certified by human checkers with a high degree of confidence?   There is always the potential that some case has been overlooked.  Indeed, the original proof by Appel and Haken apparently overlooked some cases.  Similarly, the original proof by Cilleruelo \& Luca on sums of palindromes \cite{Cilleruelo&Luca:2016} had some minor flaws that became apparent once their method was implemented as a {\tt python} program.

Third, confidence in the correctness of the results can be improved by providing code that others may check.  Transparency is essential.  To this end, we have provided our code for the nested-word automata, and the reader can easily run this code on the software we referenced.

\section{Future work}
\label{future}

This is a preliminary report.  In a later version of the paper, we hope to 
improve our bounds for generalized palindromes and antipalindromes. 

\subsection{Schnirelmann density}

Another, more number-theoretic, approach to the problems we discuss in this paper is Schnirelmann density.  Given a set $S \subseteq \Enn$, define
$A_S (x) = \sum_{{i \in S}\atop {1 \leq i \leq x}} 1$ to be the counting
function associated with $S$.  The Schnirelmann density of $S$
is then defined to be
$$ \sigma(S) := \inf_{n \geq 1}  {{A_s(n)}\over n} .$$

Classical theorems of additive number theory (e.g., 
\cite[\S 7.4]{Nathanson:1996}) relate the property of being an additive basis to the value of $\sigma(S)$.  We pose the following open problems:

\begin{openproblem}
What is the Schnirelmann density $d_k$ of those numbers expressible as the sum of at most $k$ binary palindromes?  By computation we find $d_2 <  0.443503$
and $d_3 < .942523$.
\end{openproblem}

\begin{openproblem}
Let $A$ be a $k$-automatic set of natural numbers, or its analogue using a pushdown automaton or nested-word automaton.  Is $\sigma(A)$ computable?
\end{openproblem}

\section{Moral of the story}
\label{moral}

We conclude with the following thesis, expressed as two principles.

1.  If an argument is heavily case-based, consider turning the proof into an algorithm.

2.  If an argument is heavily case-based, seek a logical system or machine model where the assertions can be expressed, and prove them purely mechanically using a decision procedure.

\medskip

Can other new results in number theory or combinatorics be proved using our approach?  We leave this as another challenge to the reader.

\section*{Acknowledgments}

We thank Dirk Nowotka and Jean-Paul Allouche for helpful discussions.  We thank the creators of the ULTIMATE automaton library for their assistance.

\end{document}